\documentclass[11pt,reqno]{amsart}
\usepackage{amsmath,amsthm,amssymb,bm,bbm,dsfont,braket}
\usepackage[mathscr]{eucal}
\usepackage{amsaddr}
\usepackage{upgreek}
\usepackage[left=2cm,right=2cm,top=2cm,bottom=2cm]{geometry}
\usepackage{mathtools}\mathtoolsset{centercolon}

\usepackage{cite}

\usepackage[shortlabels]{enumitem}
\usepackage{setspace}
\usepackage{graphicx}
\usepackage{verbatim}
\usepackage{float}
\usepackage{placeins}
\usepackage{array}
\usepackage{booktabs,arydshln}

\usepackage{threeparttable}
\usepackage[update,prepend]{epstopdf}
\usepackage{multirow}
\usepackage{amsfonts,amssymb,dsfont,xfrac,comment}
\usepackage[abs]{overpic}

\makeatletter
\def\adl@drawiv#1#2#3{%
	\hskip0
	\tabcolsep
	\xleaders#3{#2 0\@tempdimb #1{1}#2 0.5\@tempdimb}%
	#2\z@ plus1fil minus1fil\relax
	\hskip0\tabcolsep}
\newcommand{\cdashlinelr}[1]{%
	\noalign{\vskip\aboverulesep
		\global\let\@dashdrawstore\adl@draw
		\global\let\adl@draw\adl@drawiv}
	\cdashline{#1}
	\noalign{\global\let\adl@draw\@dashdrawstore
		\vskip\belowrulesep}}
\makeatother

\usepackage[usenames,dvipsnames]{xcolor}
\usepackage{colortbl}
\usepackage[hidelinks]{hyperref}
\hypersetup{
	unicode=false,          
	pdftoolbar=true,        
	pdfmenubar=true,        
	pdffitwindow=false,     
	pdfstartview={FitH},    
	pdftitle={My title},    
	pdfauthor={Author},     
	pdfsubject={Subject},   
	pdfcreator={Creator},   
	pdfproducer={Producer}, 
	pdfkeywords={keyword1} {key2} {key3}, 
	pdfnewwindow=true,      
	colorlinks=true,        
	linkcolor=Red,          
	citecolor=ForestGreen,  
	filecolor=Magenta,      
	urlcolor=BlueViolet,    
}
\usepackage{doi}
\usepackage{url}
\usepackage{caption, subcaption}
\usepackage{enumitem}
\usepackage{shadethm}
\usepackage{tikz,pgf}

\makeatletter
\ifx\@NODS\undefined%

\let\mathbb=\mathds
\else%
\fi
\makeatother

\DeclareMathOperator{\Tr}{Tr}
\DeclareMathOperator{\supp}{supp}
\DeclareMathOperator{\e}{\mathrm{e}}

\newcommand{\be}{{\mathbf e}}

\newcommand{\tr}{\operatorname{Tr}}
\newcommand{\al}{{\alpha}}

\newcommand{\pl}{\hspace{.1cm}}

\newcommand{\norm}[2]{\parallel \! #1 \! \parallel_{#2}}

\def\0{{\mathbf{0}}}
\def\1{{\mathbf{1}}}
\def\2{{\mathbf{2}}}
\def\3{{\mathbf{3}}}
\def\4{{\mathbf{4}}}
\def\5{{\mathbf{5}}}
\def\6{{\mathbf{6}}}

\def\7{{\mathbf{7}}}
\def\8{{\mathbf{8}}}
\def\9{{\mathbf{9}}}

\def\be{\begin{equation}}
\def\ee{\end{equation}}
\def\bea{\begin{eqnarray}}
\def\eea{\end{eqnarray}}

\theoremstyle{plain}
\newshadetheorem{theo}{Theorem}
\newtheorem{lemm}{Lemma} 
\newshadetheorem{prop}[theo]{Proposition} 
\newtheorem*{lemm1}{Lemma~\ref{lemm:variational}}

\theoremstyle{definition}

\theoremstyle{remark}
\newtheorem{remark}{Remark}

\makeatletter
\newcommand{\opnorm}{\@ifstar\@opnorms\@opnorm}
\newcommand{\@opnorms}[1]{%
	$\left|\mkern-1.5mu\left|\mkern-1.5mu\left|
	#1
	\right|\mkern-1.5mu\right|\mkern-1.5mu\right|$
}
\newcommand{\@opnorm}[2][]{%
	\mathopen{#1|\mkern-1.5mu#1|\mkern-1.5mu#1|}
	#2
	\mathclose{#1|\mkern-1.5mu#1|\mkern-1.5mu#1|}
}
\makeatother

\begin{document}

\let\origmaketitle\maketitle
\def\maketitle{
	\begingroup
	\def\uppercasenonmath##1{} 
	\let\MakeUppercase\relax 
	\origmaketitle
	\endgroup
}
\setlength\parindent{+4ex}
\title{\bfseries \Large{
On Strong Converse Theorems for Quantum Hypothesis Testing \\and Channel Coding
}}

\author{ \normalsize {Hao-Chung Cheng$^{1\text{--}5}$ and Li Gao$^6$}}
\address{\small  	
	$^1$Department of Electrical Engineering and Graduate Institute of Communication Engineering,\\ National Taiwan University, Taipei 106, Taiwan (R.O.C.)\\
	$^2$Department of Mathematics, National Taiwan University\\
	$^3$Center for Quantum Science and Engineering,  National Taiwan University\\
	$^4$Physics Division, National Center for Theoretical Sciences, Taipei 10617, Taiwan (R.O.C.)\\
	$^5$Hon Hai (Foxconn) Quantum Computing Center, New Taipei City 236, Taiwan (R.O.C.)\\
	$^6$School of Mathematics and Statistics, Wuhan University, Hubei Province 430072, P.~R.~China
}

\email{\href{mailto:haochung.ch@gmail.com}{haochung.ch@gmail.com}}
\email{\href{mailto:gaolimath@gmail.com}{gaolimath@gmail.com}}

\date{\today}

\begin{abstract}
Strong converse theorems refer to the study of impossibility results in information theory.
In particular, Mosonyi and Ogawa established a one-shot strong converse bound for quantum hypothesis testing [\href{https://doi.org/10.1007/s00220-014-2248-x}{\textit{Comm.~Math.~Phys}, 334(3), 2014}], which servers as a primitive tool for establishing a variety of tight strong converse theorems in quantum information theory.
In this paper, we demonstrate an alternative one-line proof for this bound via the variational expression of measured R\'enyi divergences [\href{https://doi.org/10.1007/s11005-017-0990-7}{\textit{Lett.~Math.~Phys}, 107(12), 2017}].
Then, we show that the variational expression is a direct consequence of H\"older's inequality.
\end{abstract}

\maketitle

\section{Introduction} \label{sec:intro}


\emph{Information Theory} aims to study the fundamental limits of an information-processing system and to design coding strategies to achieve the limit.
If the system is operated at certain coding rates for which the incurred error of codes $\varepsilon_n$ eventually vanishes as we increase the code length $n$, such rate is called \emph{achievable}.
The boundary of all achievable rates then corresponds to the (first-order) limit of the achievable coding rate, which is called \emph{capacity} of the channel in the contexts of channel coding.
On the other hand, \emph{converse analysis} is studying resource optimality of coding theorems.
For example, the infeasibility situation where $\liminf_{n\to +\infty}\varepsilon_n >0$ for coding rate being outside the achievable rate region is called \emph{weak converse}.
In general, there are different levels of notions of converse theorems:
\begin{itemize}
	\item
	\emph{Strong converse property}.
	The fundamental limit of achievable coding rate satisfies the strong converse property if the coding error $\varepsilon_n$ is not only bounded away from zero but converging to $1$ whenever the coding rate is operated outside the closure of the achievable rate region \cite{Sha48, Wol57}.
	
	\item
	\emph{Exponential strong converse}.
	It is called exponential strong converse if further $\varepsilon_n \to 1$ exponentially fast.
	
	\item
	\emph{Strong converse exponent}.
	To characterize how fast the coding error approaches $1$, we then call
	the largest exponential decay rate of $(1-\varepsilon_n)$ as the strong converse exponent.
	In other words, the strong converse exponent corresponds to the (first-order) limit of $- \frac1n \log (1-\varepsilon_n)$.
	
	\item
	\emph{Refined strong converse}.
	If the strong converse exponent is derived and the polynomial prefactor of the success probability is obtained as well, then we term such bounds as refined strong converses.
	This terminology was coined in Ref.~\cite{CN20}, in accordance with the refined sphere-packing bound analysis for rates within the achievable rate region	
	\cite{altugW14A,CHT19,nakiboglu19-ISIT,nakiboglu20F, AW21, Cheng2021a}.
\end{itemize}
We refer the readers to Section~\ref{sec:conclusions} for a brief historical developments of strong converse theorems.

In quantum information theory, a powerful strong converse analysis was proposed by Nagaoka~\cite{nagaoka01}, who identified Ogawa--Nagaoka's approach in quantum hypothesis testing \cite{ON00} and classical-quantum (c-q) channel coding \cite{ON99} as a consequence of the data-processing inequality of the Petz--R\'enyi divergence \cite{Pet86}.
This perspective led to a wealth of tight one-shot converse bounds in quantum information theory \cite{nagaoka01, konigW09, SW12, AMV12, MH11, Sha14, MO14, RW14, WWY14, WW14, HT14, BGW+15, guptaW15, LWD16, CMW16, tomamichelWW17, WTB17, DW18, WBH+20, Hayashi_2017, KW20, HM22, Mos23}.
In particular, Mosonyi and Ogawa employed the data-processing inequality \cite{Bei13, FL13, MDS+13, WWY14, Jencova_I, Hia21} of the sandwiched R\'enyi divergence \cite{MDS+13, WWY14} to establish that, for any test $0 \leq T \leq \mathds{1}$ \cite[Lemma 4.7]{MO14} (see also \cite[Proposition~7.71]{KW20}, \cite[Lemma 2.1]{Mos23}):
\begin{align} \label{eq:MO14}
	\Tr\left[ \rho T \right] \leq \mathrm{e}^{ - \sup_{\alpha > 1} \frac{\alpha-1}{\alpha} \left( -\log \Tr\left[ \sigma T \right] - D_{\alpha}^*(\rho\Vert \sigma ) \right) },
\end{align}
where $\Tr\left[ \rho T \right]$ is the probability of correct decision for quantum state $\rho$ using test $T$;
$\Tr\left[ \sigma T \right]$ is the probability of erroneous decision for quantum state $\sigma$ using test $T$;
and $D_{\alpha}^*(\rho\Vert \sigma)$ is the sandwiched R\'enyi divergence (see \eqref{eq:def:sandwiched} later for a detailed definition).
Mosonyi and Ogawa also proved that the established strong converse exponent is asymptotically tight \cite[Theorem 4.10]{MO14}, \cite[Theorme 4.6]{Mos23}.
The bound \eqref{eq:MO14} servers as a primitive tool to establishing tight strong converse theorems in quantum information theory not merely because this simple proof technique is widely applicable but also because it is a one-shot result and hence it holds for general situations.\footnote{We call an analysis or bound ``\emph{one-shot}'' if the underlying system does not have specific product structures and hence the statement holds for code length $n = 1$.
An asymptotic bound means that the characterization holds in the asymptotic limit $n\to +\infty$.}
By following similar reasoning of proving \eqref{eq:MO14}, the exact strong converse exponents were established for
a kind of composite quantum hypothesis testing \cite{HT14},
c-q channel discrimination \cite{WBH+20},
c-q channel coding \cite{WWY14, MO17}, constant composition coding \cite{MO18}, entanglement-assisted classical communication \cite{guptaW15, LY22},
and classical data compression with quantum side information \cite{Cheng2021a, CHDH22}.

In this paper, instead of using an information-theoretic argument such as data-processing inequality, we show that \eqref{eq:MO14} is a direct consequence of the variational expressions of quantum R\'enyi information measures \cite{FL13, BFT17, Hia21}, which can be proved using H\"older inequality about (noncommutative) $L_p$-norms (see Appendix~\ref{app:Holder}).
As demonstrated below in Theorem~\ref{theo:main}, our proof is amusingly simple and natural as well.
Our approach strikes a close connection between strong converses in quantum information theory and noncommutative functional inequalities (see also \cite{Liu18,BDR18, Cheng2021b, Esposito}).

\section{Main Results} \label{sec:main}

In Section~\ref{sec:Renyi}, we introduce the precise definition of quantum R\'enyi divergences and present their variational expressions.
In Section~\ref{sec:sc}, we show that the one-shot exponential strong converse bound is a direct consequence of the variational expressions (Theorem~\ref{theo:main}).
In Section~\ref{sec:c-q}, we apply Theorem~\ref{theo:main} to obtain a one-shot exponential strong converse bound for randomness-assisted classical-quantum channel coding, and show that shared randomness does not decrease the strong converse exponent.

\subsection{R\'enyi Information Measures and Variational Expressions} \label{sec:Renyi}
For probability mass functions $p$ and $q$, we define the (classical) order-$\alpha$ R\'enyi divergence by
\begin{align}
	D_{\alpha}(p\Vert q)
	&:=
	\begin{dcases}
	\frac{1}{\alpha-1} \log \sum\nolimits_i p(i)^\alpha q(i)^{1-\alpha} & \alpha \in (0,1)\cup (0,\infty)
	\\
	\sum\nolimits_i p(i) \log \frac{p(i)}{q(i)} & \alpha = 1
	\end{dcases}
\end{align}
if $\alpha \in (0,1)$ or $\alpha\geq1$ plus $\textrm{supp}(p) \subseteq \textrm{supp}(q)$; it is defined to be $+\infty$, otherwise \cite{Ren62, EH14}.
The order-$\infty$ R\'enyi divergence is defined via $D_{\infty}(p\Vert q) := \lim_{\alpha\to \infty} D_{\alpha}(p \Vert q)$.

For density operators $\rho$ and $\sigma$ on some separable Hilbert space, we define the order-$\alpha$ sandwiched R\'enyi divergence \cite{MDS+13, WWY14} by \cite[\S 5]{Hia21}
\begin{align}
	\label{eq:def:sandwiched}
	D_{\alpha}^*(\rho\Vert \sigma) := \frac{1}{\alpha-1} \log \Tr\left[ \left( \sigma^{\frac{1-\alpha}{2\alpha}} \rho \sigma^{\frac{1-\alpha}{2\alpha}} \right)^{\alpha} \right], \quad \alpha \in (0,1)\cup (0,\infty)
\end{align}
if $\alpha \in (0,1)$ or $\alpha>1$ plus $\textrm{supp}(\rho) \subseteq \textrm{supp}(\sigma)$; it is defined to be $+\infty$, otherwise.
The order-$1$ sandwiched R\'enyi divergences recovers Umegaki's relative entropy \cite{Ume54}
$\lim_{\alpha\to 1} D_{\alpha}^*(\rho\Vert \sigma )
 = \Tr\left[ \rho (\log \rho - \log \sigma)\right] \equiv D(\rho\Vert\sigma)$ , and order-$\infty$  is defined as $D_{\infty}^*(\rho\Vert \sigma) := \lim_{\alpha\to + \infty} D_{\alpha}^*(\rho \Vert \sigma)=\log \inf\{\lambda : \rho\le \lambda \sigma\}$.

A collection of positive semi-definite operators $\left(\Pi(i)\right)_i$ satisfying completeness relation, i.e., $\sum_i \Pi(i) = \mathds{1}$, is called positive operator-valued measure (POVM).
If $\{\Pi(i)\}_i$ are projection operators, then $\Pi$ is called projection-valued measure (PVM).
We define the projectively-measured R\'enyi divergence $D_{\alpha}^{\mathds{P}}$ and measured R\'enyi divergence $D_{\alpha}^{\mathds{M}}$ by
\begin{align}
	D_{\alpha}^{\mathds{P}}\left( \rho \Vert \sigma \right) &:= \sup \left\{ D_{\alpha}\left( p_{\rho,\Pi} \Vert p_{\sigma,\Pi} \right) : \Pi \textrm{ a PVM } \right\};
	\\
	D_{\alpha}^{\mathds{M}}\left( \rho \Vert \sigma \right) &:= \sup \left\{ D_{\alpha}\left( p_{\rho,\Pi} \Vert p_{\sigma,\Pi} \right) : \Pi \textrm{ a POVM } \right\},
\end{align}
where
\begin{align}
	p_{\rho,\Pi} = \left( p_{\rho,\Pi}(i) \right)_{i},
	\quad p_{\rho,\Pi}(i) := \Tr\left[ \rho \Pi(i) \right]
\end{align}
is the induced probability mass function for density operator $\rho$ under measurement $\Pi$.
It has been shown that $D_{\alpha}^{\mathds{P}}\left( \rho \Vert \sigma \right) = D_{\alpha}^{\mathds{M}}\left( \rho \Vert \sigma \right)$ in the finite-dimensional case \cite[Theorem 4]{BFT17} and in the infinite-dimensional case \cite[Theorem 5.8]{Hia21}.

We recall the following variational expressions, which are the main technical ingredient of the paper.
\begin{lemm}[Variational expressions {\cite{FL13, BFT17, Hia21, Mos23}}] \label{lemm:variational}
	Let $\rho$ and $\sigma$ be two density operators satisfying $\supp{(\rho)} \subseteq \supp{(\sigma)}$.
	For any order $\alpha\in (0,1)\cup (1,+\infty]$, the following identities hold:
	\begin{align}
		\label{eq:variational_measured}
		D_\alpha^{\mathbb{M}} (\rho \Vert \sigma ) &=
		\sup_{ 0 < T \leq \mathds{1} } \left\{ \frac{\alpha}{\alpha-1} \log \Tr\left[\rho T\right]
		- \log \Tr\left[ \sigma  T^{\frac{\alpha}{\alpha-1} } \right] \right\};
		\\
		\label{eq:variationa_sandwiched}
		D_\alpha^*(\rho \Vert \sigma )
		&=
		\sup_{ 0 < T \leq \mathds{1} } \left\{ \frac{\alpha}{\alpha-1} \log \Tr\left[\rho T\right]
		- \log \Tr\left[ \left(  T^{\frac12} \sigma^{\frac{\alpha-1}{\alpha}}  T^{\frac12} \right)^{\frac{\alpha}{\alpha-1}} \right] \right\}.
	\end{align}
	For $\alpha>1$, the supremum can be relaxed to $0\leq T\leq \mathds{1}$
	such that
	$T$ is not orthogonal to $\rho$.
\end{lemm}
\begin{remark}
	Both the objective functions on the right-hand sides are homogeneous in $T$.
\end{remark}

The variational expression for the sandwiched quantity $D_\alpha^*$ in \eqref{eq:variationa_sandwiched} was first proved by Frank and Lieb \cite[Lemma 4]{FL13} in the finite-dimensional case, and was extended to the infinite-dimensional case by Jen{\v{c}}ov{\'a} \cite[Proposition 3.4]{Jencova_II} and Hiai \cite[Lemma 3.19]{Hia21}.
Mosonyi established the variational expression of the $(\alpha,z)$-R\'enyi divergence \cite{AD15} in the infinite-dimensional case, which also covers the sandwiched quantity \cite[Lemma 3.23]{Mos23}.
The variational expression for the measured quantity $D_\alpha^{\mathbb{M}}$ in \eqref{eq:variational_measured} was first proved in the finite-dimensional case by Berta, Fawzi, and Tomamichel in \cite[Lemma 3]{BFT17}, and generalized to the infinite-dimensional case by Hiai \cite[(5.17), (5.18)]{Hia21}.
Note that in the above references, the upper bound $T\leq \mathds{1}$ in the optimization set was not mentioned.

The inequality $D_\alpha^{\mathbb{M}} (\rho \Vert \sigma )\leq D_\alpha^*(\rho \Vert \sigma ) $ for all $\alpha\geq \frac12$ can be seen from the data-processing inequality of the sandwiched R\'enyi divergence \cite{Bei13, FL13, MDS+13, Jencova_I, Hia21, Mos23}.
Moreover, the variational expressions directly show that the inequality is strict for $\alpha \in (\sfrac12, \infty)$ unless $\rho \sigma = \sigma\rho$ or both the quantities are infinite via the Araki--Lieb--Thirring inequality \cite{LT76, Ara90, Hia94}.

For completeness, we will provide an alternative proof of the variational expressions (for separable Hilbert spaces) in Appendix~\ref{app:Holder}. They are both direct consequences of H\"older's inequality.

\begin{remark}
	For $\alpha=1$, Petz has proved the following variational expressions \cite{Don86, Pet88, Pet08} (see also \cite{BFT17}):
	\begin{align}
		D_1^{\mathds{M}}(\rho\Vert \sigma)
		&= \sup_{0<T\leq \mathds{1}} \left\{ \Tr\left[ \rho \log T \right] - \log \Tr\left[ \sigma T \right]
		\right\};
		\\
		D_1^*(\rho\Vert \sigma)
		= D(\rho\Vert \sigma)
		&= \sup_{0<T\leq \mathds{1}} \left\{ \Tr\left[ \rho \log T \right] - \log \Tr\left[ \exp\{ \log \sigma + \log T  \} \right]
		\right\}.
	\end{align}
\end{remark}

\subsection{Strong Converses via Variational Expressions} \label{sec:sc}

Consider testing two hypotheses of quantum states $\rho$ and $\sigma$ (both are mathematically described by some density operators on a separable Hilbert space).
Given a test $0\leq T\leq \mathds{1}$ for deciding the null hypothesis $\rho$, the
probability of correct decision under the true hypothesis $\rho$ is
$\Tr\left[\rho T \right]$ (which is also called the Type-I success probability), and the probability of erroneous decision under the alternative hypothesis $\Tr\left[ \sigma T \right]$ (which is also called the type-II error probability).
Apparently, if one desires the type-II error $\Tr\left[\sigma T \right]$ to be small, then the type-I success probability $\Tr\left[\rho T \right]$ cannot be too large.
The strong converse exponent analysis is hence to find the optimal trade-off between these two quantities.

\begin{theo}[One-shot exponential strong converse] \label{theo:main}
	Let $\rho$ and $\sigma$ be density operators satisfying $D_1^{\mathds{M}}(\rho\Vert\sigma)<+\infty$.
	Then, for any test $0\leq T \leq \mathds{1}$ non-orthogonal to $\rho$, 
	\begin{align}
		\Tr\left[ \rho  T \right]
		\leq \e^{ -  \frac{\alpha-1}{\alpha} \left( -\log \Tr\left[ \sigma  T \right] - D_\alpha^{\mathds{M}}(\rho\Vert \sigma) \right) },
		\quad \forall\, \alpha\geq 1.
	\end{align}
\end{theo}
\begin{proof}
	We only consider $\alpha>1$ with $D_{\alpha}^{\mathds{M}}(\rho\Vert\sigma)<+\infty$ (if there is any) because otherwise the bound holds trivially.
	Then, we must have $\textrm{supp}(\rho)\subseteq \textrm{supp}(\sigma)$.
	The variational formula, \eqref{eq:variational_measured} of Lemma~\ref{lemm:variational}, implies
	\begin{align}
		D_\alpha^{\mathbb{M}} (\rho \Vert \sigma )
		&= \sup_{0\leq T\leq \mathds{1} } \left\{ \frac{\alpha}{\alpha-1} \log \Tr\left[\rho T\right]
		- \log \Tr\left[ \sigma  T^{\frac{\alpha}{\alpha-1} } \right] \right\}
		\\
		&\geq \sup_{0\leq T\leq \mathds{1} } \left\{ \frac{\alpha}{\alpha-1} \log \Tr\left[\rho T\right]
		- \log \Tr\left[ \sigma  T  \right] \right\},
	\end{align}
	where the inequality is because $\frac{\alpha}{\alpha-1}\geq1$ and hence $T^{\frac{\alpha}{\alpha-1} } \leq T \leq \mathds{1}$.
	This concludes the proof.
\end{proof}

One can further relax the right-hand side to obtain an additive exponent, i.e.~$D_\alpha^{\mathds{M}}(\rho\Vert \sigma) \leq D_\alpha^{*}(\rho\Vert \sigma)  $
\cite{Bei13, FL13, MDS+13, BFT17, Jencova_I, Hia21, Mos23}.
This then yields a lower bound to the strong converse exponent for quantum hypothesis testing for $n$-fold product states: 
given any integer $n$,
\begin{align}
	\sup_{0\leq T^n \leq \mathds{1}^{\otimes n}
	}
	\left\{
	- \frac{1}{n} \log \Tr\left[ \rho^{\otimes n} T^n  \right]
	: -\frac{1}{n}\log  \Tr\left[ \sigma^{\otimes n} T^n  \right] \geq r
	\right\}
	&\geq
	\sup_{\alpha\geq1} \frac{\alpha-1}{\alpha} \left( r -  D_{\alpha}^{*}(\rho\Vert \sigma)    \right),
\end{align}
where the exponent is positive if and only if $r> \lim_{\alpha \searrow1 } D_{\alpha}^{*}(\rho\Vert \sigma) = D(\rho\Vert \sigma)$.
The reverse inequality (i.e.~upper bound) was shown in the asymptotic limit $n\to +\infty$ by Mosonyi and Ogawa \cite[Theorem 4.10]{MO14}, \cite[Theorem 4.6]{Mos23}.

On the other hand, one can also directly evaluate the regularized measured R\'enyi divergence in Theorem~\ref{theo:main} in the asymptotic scenario.
Indeed, it was shown that \cite{HP91, HT14}, \cite[Corollary 4.6]{MO14}, \cite[(4.34)]{hiai2017different},
\begin{align}
	\lim_{n\to \infty} \frac{1}{n} D_\alpha^{\mathds{M}}\left( \rho^{\otimes n} \Vert \sigma^{\otimes n}\right)
	= D_{\alpha}^*(\rho\Vert \sigma),
	\quad \alpha\geq \frac12
\end{align}
for finite-dimensional Hilbert spaces,
and it was later generalized to the infinite-dimensional Hilbert space by Mosonyi \cite[Theorem 3.45]{Mos23} and semi-finite von Neumann algebra in \cite{fawzi2022asymptotic}.
Nevertheless, from the proof of Theorem~\ref{theo:main}, it seems that the measured R\'enyi divergence $D_{\alpha}^{\mathds{M}}$ is a more natural and tight quantity in the one-shot bound.

\subsection{Randomness-Assisted Classical-Quantum Channel Coding} \label{sec:c-q}

Consider a classical-quantum (c-q) channel $\mathscr{N}_{\mathsf{X} \to \mathsf{B}} : x\mapsto \rho_{\mathsf{B}}^x$, which maps each letter $x$ in a finite alphabet $\mathsf{X}$ to a density operator $\rho_{\mathsf{B}}^x$ on some output Hilbert space $\mathcal{H}_{\mathsf{B}}$.
The goal of a c-q channel coding is to send equiprobable messages $m\in \mathsf{M}:= \{1,2,\ldots, |\mathsf{M}|\}$ over the channel using proper designed encoding and decoding.
As we are considering the strong converse regime, we further grant \emph{unlimited public randomness} shared between encoder at Alice and decoder at Bob.
We call this \emph{randomness-assisted classical communication over $\mathscr{N}_{\mathsf{X} \to \mathsf{B}}$} (see e.g.~\cite[\S 3.2]{Wil17b}), and shortly we will show impossibility result even with assistance of shared randomness.
A randomness-assisted code is described as follows.
\begin{itemize}
	\item \emph{Preparation.} Before communication commences, Alice and Bob prepare shared randomness with a common probability distribution $p_{\mathsf{X}}$ on the input alphabet $\mathsf{X}$.
	
	\item \emph{Encoding}.
	For each message $m\in \mathsf{M}$, Alice chooses a codeword $x(m) \in \mathsf{X}$ via
	an encoder $\mathscr{E}:\mathsf{M}\to \mathsf{X}$, and sends it to the channel $\mathscr{N}_{\mathsf{X} \to \mathsf{B}}$.
	
	\item \emph{Decoding.}
	Via the shared randomness, Bob chooses a positive operator-valued measure (POVM), i.e.~$\left(T_{\mathsf{B}}^{x(m)} \right)_{m\in\mathsf{M}}$ with $T_{\mathsf{B}}^{x(m)}\geq 0$ and $\sum_{m\in\mathsf{M}} T_{\mathsf{B}}^{x(m)} = \mathds{1}_{\mathsf{B}}$ as a decoder $\mathscr{D}$, to decode the sent message $m$.
\end{itemize}
The largest probability of successful decoding (with shared randomness $p_{\mathsf{X}}$) is hence given by
\begin{align}
	P(\mathsf{X}:\mathsf{B})_{\rho}
	:= \sup_{ (\mathscr{E}, \mathscr{D})  }
	\frac{1}{|\mathsf{M}|} \sum_{m\in\mathsf{M}} \mathds{E}_{x(m)\sim p_{\mathsf{X}}} \Tr\left[ \rho_{\mathsf{B}}^{x(m)} T_{\mathsf{B}}^{x(m)} \right].
\end{align}
By applying Theorem~\ref{theo:main}, we have the following bound.

\begin{prop}[One-shot exponential strong converse for classical-quantum channel coding] \label{prop:c-q}
	Consider any classical-quantum channel $\mathscr{N}_{\mathsf{X} \to \mathsf{B}}: x\mapsto \rho_{\mathsf{B}}^x$.
	Any randomness-assisted codes with common distribution $p_{\mathsf{X}}$ for sending equiprobable $|\mathsf{M}|$ messages over $\mathscr{N}_{\mathsf{X} \to \mathsf{B}}$ satisfies
	\begin{align}
		P(\mathsf{X}:\mathsf{B})_{\rho}
		\leq \mathrm{e}^{ - \frac{\alpha-1}{\alpha} \left( \log |\mathsf{M}| -  I_{\alpha}^{\mathds{M}}( \mathsf{X} : \mathsf{B} )_{\rho}   \right) },
		\quad \forall\, \alpha \geq1,
	\end{align}
	where $\rho_{\mathsf{XB}} := \mathds{E}_{x\sim p_{\mathsf{X}}} |x\rangle\langle x|_{\mathsf{X}} \otimes \rho_{\mathsf{B}}^{x}$,
	and $I_{\alpha}^{\mathds{M}}( \mathsf{X} : \mathsf{B} )_{\rho} := \inf_{\sigma_{\mathsf{B}}} D_{\alpha}^{\mathds{M}}(\rho_{\mathsf{XB}} \Vert \rho_{\mathsf{X}} \otimes \sigma_{\mathsf{B}} )$ is the measured R\'enyi information (with minimization over density operators $\sigma_{\mathsf{B}}$ on Hilbert space $\mathcal{H}_{\mathsf{B}}$).
\end{prop}

Noting that the established bound holds for any distribution $p_{\mathsf{X}}$, one may upper bound $I_{\alpha}^{\mathds{M}}( \mathsf{X} : \mathsf{B} )_{\rho}$ on the  right-hand side to obtain the measured R\'enyi capacity and further relax it to the sandwiched R\'enyi capacity (see e.g.~\cite[\S 4]{MO17}, \cite{HT14}, \cite[Proposition 4]{CGH18}):
\begin{align}
	\sup_{p_{\mathsf{X}}} I_{\alpha}^{\mathds{M}}( \mathsf{X} : \mathsf{B} )_{\rho}
	=: C_{\alpha}^{\mathds{M}}(\mathscr{N}_{\mathsf{X} \to \mathsf{B}})
	\leq C_{\alpha}^{*}(\mathscr{N}_{\mathsf{X} \to \mathsf{B}})
	:= \sup_{p_{\mathsf{X}}} \inf_{\sigma_{\mathsf{B}}} D_{\alpha}^{*}(\rho_{\mathsf{XB}} \Vert \rho_{\mathsf{X}} \otimes \sigma_{\mathsf{B}} )
	\leq \log |\mathsf{X}|.
\end{align}
In the setting of $n$-fold product channels  $\mathscr{N}_{\mathsf{X} \to \mathsf{B}}^{\otimes n}$, invoking the additivity, i.e.~$C_{\alpha}^{*}(\mathscr{N}_{\mathsf{X} \to \mathsf{B}}^{\otimes n}) = n C_{\alpha}^{*}(\mathscr{N}_{\mathsf{X} \to \mathsf{B}})$ \cite[Lemma 6]{guptaW15}, \cite{Bei13, DJK+06, WWY14}, we hence obtain a lower bound to the strong converse exponent for randomness-assisted coding over $\mathscr{N}_{\mathsf{X} \to \mathsf{B}}^{\otimes n}$ with rate $R:= \frac1n \log |\mathsf{M}|$ (for any integer $n$) by
\begin{align} \label{eq:sc_exponent_c-q}
	\inf_{p_{\mathsf{X}^n}} -\frac1n \log P(\mathsf{X}^n:\mathsf{B}^n)_{\rho^n} \geq 
	\sup_{\alpha\geq 1} \frac{\alpha-1}{\alpha} \left( R -  C_{\alpha}^{*}(\mathscr{N}_{\mathsf{X} \to \mathsf{B}})    \right).
\end{align}
This means that the probability of successful decoding, $P(X^n\!:\!B^n)_{\rho^n}$, with any shared randomness $p_{\mathsf{X}^n}$ vanishes for any blocklength $n$ with an exponent at least the right-hand side of \eqref{eq:sc_exponent_c-q}. (Here, the joint state $\rho_{\mathsf{X}^n \mathsf{B}^n}^n$ is induced from $p_{\mathsf{X}^n}$ and $\mathscr{N}_{\mathsf{X} \to \mathsf{B}}^{\otimes n}$.)
Since \eqref{eq:sc_exponent_c-q} can be asymptotically (as $n\to +\infty$) achieved by codes without randomness assistance \cite[Theorem 5.14]{MO17}, the above strong converse exponent is exact, and we conclude that shared randomness does not \emph{decrease} the strong converse exponent for communication over $\mathscr{N}_{\mathsf{X} \to \mathsf{B}}$.
Namely, shared randomness does not help in the strong converse rate region $R> C_1^*(\mathscr{N}_{\mathsf{X} \to \mathsf{B}})$.\footnote{It is well-known that shared randomness does not increase the channel capacity, $C_1^*(\mathscr{N}_{\mathsf{X} \to \mathsf{B}})$, of a classical-quantum channel $\mathscr{N}_{\mathsf{X} \to \mathsf{B}}$.
We then further show that shared randomness does not help to slow the exponential decay rate of the probability of successful decoding.
}

We shall emphasize that Proposition~\ref{prop:c-q} can be proved by using the standard argument of the data-processing inequality as shown in \cite[\S 4]{WWY14}, \cite[Lemma 5.5]{MO17} (see also \cite[Lemma 31]{nakiboglu19B}).
Below we show that Proposition~\ref{prop:c-q} is also a direct consequence of Theorem~\ref{theo:main} without resorting to data-processing inequalities.

\begin{proof}[Proof of Proposition~\ref{prop:c-q}]
	For any codes, we define the following
	\begin{align}
		\omega_{\mathsf{MXB}}
		&:= \frac{1}{|\mathsf{M}|} \sum_{m\in\mathsf{M}} |m\rangle\langle m|_{\mathsf{M}}
		\otimes \sum_{x(m)\in\mathsf{X}} p_{\mathsf{X}}(x(m)) |x(m)\rangle \langle x(m)|_{\mathsf{X}} \otimes \rho_{\mathsf{B}}^{x(m)};
		\\
		T_{\mathsf{MXB}}
		&:= \sum_{m\in\mathsf{M}} |m\rangle\langle m|_{\mathsf{M}}
		\otimes \sum_{x(m)\in\mathsf{X}} |x(m)\rangle \langle x(m)|_{\mathsf{X}} \otimes T_{\mathsf{B}}^{x(m)}.
	\end{align}
	By inspection, we obtain for any density operator $\sigma_{\mathsf{B}}$,
	\begin{align}
		\Tr\left[ \omega_{\mathsf{MXB}} T_{\mathsf{MXB}}  \right]
		&= P(X:B)_{\rho};
		\\
		\Tr\left[ \omega_{\mathsf{MX}} \otimes \sigma_{\mathsf{B}} T_{\mathsf{MXB}}  \right]
		&= \frac{1}{|\mathsf{M}|}.
	\end{align}
	
	Now we apply Theorem~\ref{theo:main} with $\rho \leftarrow \omega_{\mathsf{MXB}}$, $\sigma \leftarrow \omega_{\mathsf{MX}} \otimes \sigma_{\mathsf{B}}$
	and $T \leftarrow T_{\mathsf{MXB}} $, and identify that
		\begin{align}
			D_{\alpha}^{\mathds{M}}\left( \omega_{\mathsf{MXB}} \Vert \omega_{\mathsf{MX}} \otimes \sigma_{\mathsf{B}} \right)
			= D_{\alpha}^{\mathds{M}}(\rho_{\mathsf{XB}} \Vert \rho_{\mathsf{X}} \otimes \sigma_{\mathsf{B}} )
		\end{align}
		by the direct-sum structure.
	By minimizing over $\sigma_{\mathsf{B}}$, we conclude the proof.
\end{proof}

\begin{remark}
	Similar reasoning as shown in Proposition~\ref{prop:c-q} directly leads to
	the strong converse exponent for position-based entanglement-assisted classical communication over quantum channels (see the setup in e.g.~\cite[\S IV.B]{Cheng_simple} and \cite{Wil17b, AJW19a}), firstly proved in \cite{guptaW15}, where the tightness was recently proved in \cite{LY22}.
\end{remark}

\section{Historical Remarks and Conclusions} \label{sec:conclusions}

The converse study can be traced back to the origin of Information Theory.
In Shannon's pioneering paper \cite{Sha48}, the strong converse property of channel capacity was mentioned, but the first proof was given by Wolfowitz via a type counting method \cite{Wol57, Wol59, Wol63, Wol64, Wol68} (see also \cite{kemperman69A, kemperman69B, HK89, Kem71, WIK12}, \cite[\S 5.3]{Ahl14}).
Subsequently, strong converse analysis becomes a rich direction of research;
various strong converse theorems have been proved and applied to numerous information-theoretic tasks, e.g.,
a combinatorial approach called the blowing-up lemma \cite{marton1966simple, AD76, AGK76}, \cite[\S 5]{CK11}, \cite[\S 3.6]{RS13},
Csisz{\'{a}}r--K{\"{o}}rner's entropy and image size characterization \cite[\S 15]{CK11},
Augustin's non-asymptotic converse bound \cite[\S 10]{augustin66}, \cite[\S 13]{augustin78},
Dueck's wringing technique \cite{Due81, Ahl82, FT16},
Han's information-spectrum method \cite{verduH94, Han98, Han03, BS00, Hay09b, WH14, TB15} and later refined by Oohama \cite{OH94, Ooh15_IEICE, Ooh15a, Ooh15b, Ooh16, oohama17A, Ooh18, Ooh19, Ooh19b, Ooh20, SO19, OS22},
Polyanskiy--Poor--Verd{\'u}'s meta-converse of reducing converse problems to hypothesis testing \cite{PPV10}, \cite[\S 4]{Tan14}, \cite{KV12, KV13, HTW14, VCF+16},
strong converse via functional inequalities such as Gaussian Poincar{\'e} inequality \cite{FT17},
reverse hypercontractivity \cite{Liu18a, Liu18, LHV17, LHV18}, and Brascamp-Lieb Inequalities \cite{LCV20},
and a change-of-measure argument \cite{GE09, GE11, Wat17, TW18, TW23}.

In particular, Arimoto established a {one-shot strong converse bound} of the form in \eqref{eq:MO14} for discrete classical channels \cite[Theorem 1]{Ari73}.
If the channel is stationary and product, it leads to exponential convergence of $(1-\varepsilon_n)$ for any code length $n$.
Omura \cite{omura75}, and then Dueck and K{\"o}rner \cite{dueckK79} showed that Arimoto's strong converse exponent is tight as $n\to +\infty$.
Later, Arimoto's bound based on Jensen's inequality was formulated into an equivalent result of the data-processing inequality of R\'enyi divergences \cite{Ren62} by Augustin \cite[Theorem 27.2-(ii)]{augustin78} and Sheverdyaev \cite[Lemma 2]{sheverdyaev82}, respectively.
Polyanskiy and Verd{\'u} extended this approach to any generalized divergence satisfying data-processing inequality \cite{PV10a}.
(We refer the readers to a excellent detailed historical discussion by Nakibo{\u{g}}lu in \cite[Appendix B]{nakiboglu19B}.)
Refined strong converse theorems for hypothesis testing and channel coding with explicit polynomial prefactors were established in \cite{Str62, csiszarL71, vazquezFKL18, CN20}.

Extending the ideas of classical strong converse theorems to the quantum scenario is non-trivial and facing several technical challenges due to the noncommutative nature of quantum information.
After a series of efforts, substantial progress has been made and novel techniques have been developed along this line of research.
For quantum hypothesis testing, Ogawa and Nagaoka analyzed the quantum Neyman--Pearson test \cite{Hel67, Hol72}, \cite[p.~120]{Hel76}, \cite[Theorem 3.4]{Wat18} to obtain a one-shot exponential strong converse bound \cite{ON00}, completing the strong converse part of the famous Quantum Stein's lemma by Hiai and Petz \cite{HP91,ON00}, and also providing a lower bound to the {strong converse exponent}.
In addition, Ogawa and Nagaoka combined Arimoto's approach \cite{Ari73} with techniques in matrix analysis to obtain both one-shot and exponential strong converses for c-q channel coding \cite{ON99}, completing the strong converse of the Holevo--Schumacher--Westmoreland (HSW) theorem for c-q channel coding \cite{Hol73b, SW97, Hol98}.
Independently, Winter extended the type-counting method with a gentle measurement lemma \cite[Lemma 9]{Win99b} to obtain strong converses of Wolfowitz's form \cite{Wol57} for non-stationary product c-q channels and also for quantum compression \cite{Win99} (see also \cite{Sch95, JS94, BFJ+96, BCF+01}).
Hayashi proved a lower bound to the strong converse exponent for quantum compression of pure-state ensembles based on certain techniques in entanglement theory \cite{NK01, Hay02b}.
Hayashi and Nagaoka in \cite[Lemma 4]{HN03} established Verd{\'u}--Han's one-shot converse bound for general c-q channels \cite[Theorem 4]{VH94}.
Ahlswede and Cai applied the wringing technique to show the strong converse property for c-q multiple-access channels \cite{AC05}.
An interesting strong converse analysis for channel identification and entanglement-assisted communication were established via the covering lemma \cite{Win02, AW02} and channel simulation \cite[\S IV-E]{6757002}, \cite{BCR11, BBC+13}.
Another widely useful converse analysis is based on the \emph{smooth-entropy framework} \cite{tom-thesis, TH13}.
Together with the {Quantum Asymptotic Equipartition Property} \cite{TCR09} or the chain rules of the smoothed entropies \cite{TCR10, DBW+14}, numerous strong converse theorems and bounds were proved \cite{DBW+14, DMH+13, WW14, MW14_pretty_strong, Win16}.
Polyanskiy--Poor--Verd{\'u}'s meta-converse was extended to quantum coding by Wang--Renner \cite{WR13} and Matthews-Wehner \cite{MW14}.\footnote{We note that the concept of meta-converse might have  appeared in \cite[Lemma 4]{HN03}, which is  expressed in the form of the \emph{information-spectrum divergence} \cite[(2.9)]{Tan14} instead of the \emph{hypothesis-testing divergence} \cite[(2.6)]{Tan14}, \cite[(1)]{WR13}, \cite[(22)]{MW14}.
}
On the other hand, Li \cite{Li14} and Tomamichel--Hayashi \cite{TH13} established the second-order asymptotics of the hypothesis-testing divergence independently using different techniques; the above implies an asymptotic exponential strong converse for quantum hypothesis testing.
Combined with the meta-converse, second-order converse bounds for certain quantum channels were obtained \cite{TV15, WRG15, DTW16, TBR16, WTB17, CWY23}.
Strong converse bounds for quantum channels via semidefinite programming were developed by Wang \textit{et al.} \cite{WXD18, WKD19}.
Strong converse bounds were established for various quantum network information-theoretic tasks via quantum reverse hypercontractivity \cite{BDR18, Cheng2021b}.
Strong converse theorems for privacy amplification against quantum side information, classical-quantum soft covering, and quantum information decoupling were recently shown via more involved analysis \cite{KL21, SD22, LY22, LY24, CG22, SGC22a, SGC22b, SGC23, CG23}.

As mentioned in Section~\ref{sec:intro}, Nagaoka in Ref.~\cite{nagaoka01} identified Ogawa--Nagaoka's converse analysis in quantum hypothesis testing \cite{ON00} and c-q channel coding \cite{ON99} as a consequence of the data-processing inequality of the Petz--R\'enyi divergence \cite{Pet86}, aligned with the observations made by Augustin \cite{augustin78} and Sheverdyaev \cite{sheverdyaev82}.\footnote{It seems that Nagaoka in Ref.~\cite{nagaoka01} independently discovered this fact in the quantum setting without being aware of Augustin's and Sheverdyaev's works \cite{augustin78, sheverdyaev82}.}
This perspective along with K\"onig--Wehner's study \cite{konigW09}, Polyanskiy--Verd{\'u}'s generalization \cite{PV10a} led to a wealth of tight one-shot converse bounds in quantum information theory \cite{nagaoka01, konigW09, SW12, AMV12, MH11, Sha14, MO14, RW14, WWY14, WW14, HT14, BGW+15, guptaW15, LWD16, CMW16, tomamichelWW17, WTB17, DW18, WBH+20, Hayashi_2017, KW20, HM22, Mos23} as described in Introduction.
In particular, Mosonyi--Ogawa's strong converse exponent for quantum hypothesis testing \cite{MO14} can be viewed as a quantum analogue of the Han--Kobayashi bound in classical hypothesis testing \cite{HK89}.\footnote{The strong converse exponent in the Han--Kobayashi bound \cite{HK89} is expressed in terms of Blahut's form \cite{Bla74} (similar to the Haroutunian form \cite{haroutunian68} in channel coding), while \eqref{eq:MO14} is expressed in terms of the parametric form of R\'enyi divergences.
	We note that in the quantum setting, Blahut's form corresponds to the so-called log-Euclidean R\'enyi divergence that is not equal to the sandwiched R\'enyi divergence; see e.g., \cite[Lemma 5]{CHT19}, \cite[Propositions 3.8 \& 3.20]{MO17}, \cite[\S 3]{Hayashi_2017}.
}
Hence, Mosonyi and Ogawa's one-shot result in \eqref{eq:MO14} servers as a fundamental tight bound for strong converse analysis in quantum information.

In this paper, we have provided an alternative simple proof of \eqref{eq:MO14} via the variational expression of the measured R\'enyi divergence; essentially, it is only a one-step proof.
Our argument may provide a new insight on this fundamental strong converse bound.
Although the variational expressions are already known, yet in Appendix~\ref{app:Holder}, we provide a streamlined way of proving the variational expressions for the sandwiched R\'enyi divergence and the measured R\'enyi divergence.
Namely, we show that how both the variational expressions can be equivalently unified as a type of H\"older's inequality.

\section*{Acknowledgment}
H.-C.~Cheng would like to thank Bar\i\c{s} Nakibo\u{g}lu for his always help and insightful discussions.
H.-C.~Cheng is supported by the Young Scholar Fellowship (Einstein Program) of the National Science and Technology Council, Taiwan (R.O.C.) under Grants No.~NSTC 112-2636-E-002-009, No.~NSTC 113-2119-M-007-006, No.~NSTC 113-2119-M-001-006, No.~NSTC 113-2124-M-002-003, by the Yushan Young Scholar Program of the Ministry of Education, Taiwan (R.O.C.) under Grants No.~NTU-112V1904-4 and by the research project ``Pioneering Research in Forefront Quantum Computing, Learning and Engineering'' of National Taiwan University under Grant No. NTU-CC- 112L893405 and NTU-CC-113L891605. H.-C.~Cheng acknowledges the support from the “Center for Advanced Computing and Imaging in Biomedicine (NTU-113L900702)” through The Featured Areas Research Center Program within the framework of the Higher Education Sprout Project by the Ministry of Education (MOE) in Taiwan.

\appendix
\section{Variational Expressions via H\"older's inequality} \label{app:Holder}

In this section, we show that the variational expressions given in Lemma~\ref{lemm:variational} are actually equivalent to certain H\"older inequality.
Before commencing, let us first introduce necessary notations.

Let $\mathcal{H}$ be a separable Hilbert space,
$\mathcal{B(H)}$ be the space of bounded operators on $\mathcal{H}$,
and $Y\in \mathcal{B(H)}$ be an positive definite operator. Here, positive definite means that $\langle h | Y |h \rangle>0 $ for all nonzero $h \in \mathcal{H}$.
We define the \emph{noncommutative quotient} of an operator $X$ over $Y$ as:
\begin{align}
	\frac{X}{Y} := Y^{-\frac12 } X Y^{-\frac12 }.
\end{align}
For any positive operator $\sigma\in \mathcal{B(H)}$, the weighted $p$-norm of $X \in \mathcal{B(H)}$ with respect to $\sigma$ is defined as \cite{Kos84}:\footnote{For $p< 1$, the quantity defined in \eqref{eq:weighted-p-norm} is not a norm, but for simplicity we still adopt such an expression.}
\begin{align} \label{eq:weighted-p-norm}
	\left\|X\right\|_{p,\sigma} := \left( \Tr\left[ \left| \sigma^{\frac{1}{2p}} X \sigma^{\frac{1}{2p}} \right|^p \right] \right)^{\frac{1}{p}},
	\quad \forall\, p \in \mathds{R}.
\end{align}
We denote $\left\|X\right\|_{0,\sigma} := \lim_{p\to 0} \left\|X\right\|_{p,\sigma}$ and $\left\|X\right\|_{\infty,\sigma} := \lim_{p\to +\infty} \left\|X\right\|_{p,\sigma}=\left\|X\right\|_{\infty} $ which is the operator norm of $X$. The corresponding $L_p$ space we denote as that $L_p(\mathcal{B(H)},\sigma)$.
Recall the Kubo--Martin--Schwinger (KMS) $\sigma$-weighted inner product
\begin{align}
	\left\langle X, Y \right\rangle_{\sigma} := \Tr\left[ X^\dagger \sigma^{\frac12} Y \sigma^{\frac12} \right].
\end{align}
The weighted norm and inner product satisfies the following H\"older's inequality.

\begin{theo}[H\"older's inequality] \label{theo:Holder}
	Let $\alpha \in (0,\infty]$ and let $\alpha':= \frac{\alpha}{\alpha-1}$ be its conjugate index.
	Then, for all positive semi-definite $X$ and $Y$, the following holds:
	\begin{align}
		\label{eq:Holder_variant}
		\left\langle X, Y \right\rangle_{\sigma}
		\leq \left\| X \right\|_{\alpha, \sigma } \cdot \left\| Y \right\|_{\alpha', \sigma },
		\quad \alpha \geq 1;
	\end{align}
and if in additional, $\mathrm{supp}(\sigma) \subseteq \mathrm{supp}(Y)$,
\begin{align}
		\label{eq:Holder_variant_reverse}
		\left\langle X, Y \right\rangle_{\sigma}
		\geq
		\left\| X \right\|_{\alpha, \sigma } \cdot \left\| Y \right\|_{\alpha', \sigma },
		 \quad \alpha \in (0,1) \;.
	\end{align}
	
If $\alpha \in (1,\infty)$ and $\left\| X \right\|_{\alpha, \sigma } \cdot \left\| Y \right\|_{\alpha', \sigma } < +\infty$,
	the equality is attained if and only if
	\begin{align} \label{eq:equality}
		\left| \sigma^{\frac{1}{2\alpha}} X \sigma^{\frac{1}{2\alpha}} \right|^{\alpha}
		=
		c \left| \sigma^{\frac{1}{2\alpha'}} Y \sigma^{\frac{1}{2\alpha'}} \right|^{\alpha'}
	\end{align}
	for some scalar $c\geq 0$.
	
	On the other hand, if $\alpha \in (0,1)$, $\left\langle X, Y \right\rangle_{\sigma}<+\infty$, and $\left\| Y \right\|_{\alpha', \sigma }> 0$,
	the equality is attained if and only if
	\eqref{eq:equality} holds for some scalar $c\geq 0$.
\end{theo}

\begin{proof}
	For $\alpha>1$, let
	\begin{align}
		\tilde{X} = \sigma^{\frac{1}{2\alpha}} X \sigma^{\frac{1}{2\alpha}}\pl, \pl
		\tilde{Y} = \sigma^{\frac{1}{2\alpha'}} Y \sigma^{\frac{1}{2\alpha'}}.
	\end{align}
	Then, \eqref{eq:Holder_variant} follows from the standard trace H\"older's inequality for $\tilde{X}$ and $\tilde{Y}$ (see e.g., \cite[Corollary IV.2.6]{Bha97}, \cite[Theorem 6.20]{HP14}, \cite[Lemma 6]{TBH14}, \cite[Theorem A.38]{Hia21}):
	\begin{align}
		\Tr\left[ \tilde{X}
		\tilde{Y} \right]
		\leq \|\tilde{X}\|_{\alpha} \|\tilde{Y}\|_{\alpha'},
	\end{align}
	where $\|\tilde{X}\|_{\alpha} := \left( \Tr[|\tilde{X}|^\alpha] \right)^{1/\alpha} $ is the Schatten $\alpha$-norm. Also the equality condition follows from the \cite[Proposition 8]{dixmier1953formes}.
	
	For $\alpha \in (0,1)$ and 
	$\alpha'\in (-\infty,0)$,
	the reverse H\"older's inequality already appeared in	\cite[Lemma 1]{BDR18}, which by following the same reasoning holds because of the standard H\"older's inequality
	\cite[Lemma 6]{TBH14}. The same argument works for infinite-dimensional Hilbert spaces. 
\end{proof}

Now we ready to provide a simple proof of the variational expressions via H\"older's inequality, which we restate below. 

\begin{lemm1}[Variational expressions {\cite{FL13, BFT17, Hia21, Mos23}}] 
	Let $\rho$ and $\sigma$ be two density operators satisfying $\supp{(\rho)} \subseteq \supp{(\sigma)}$.
	For any order $\alpha\in (0,1)\cup (1,+\infty]$, the following identities hold:
\begin{align}
	\label{eq:variational_measured_app}
	D_\alpha^{\mathbb{M}} (\rho \Vert \sigma ) &=
	\sup_{ 0 < T \leq \mathds{1} } \left\{ \frac{\alpha}{\alpha-1} \log \Tr\left[\rho T\right]
	- \log \Tr\left[ \sigma  T^{\frac{\alpha}{\alpha-1} } \right] \right\};
	\\
	\label{eq:variationa_sandwiched_app}
	D_\alpha^*(\rho \Vert \sigma )
	&=
	\sup_{ 0 < T \leq \mathds{1} } \left\{ \frac{\alpha}{\alpha-1} \log \Tr\left[\rho T\right]
	- \log \Tr\left[ \left(  T^{\frac12} \sigma^{\frac{\alpha-1}{\alpha}}  T^{\frac12} \right)^{\frac{\alpha}{\alpha-1}} \right] \right\}.
\end{align}
For $\alpha>1$, the supremum can be relaxed to $0\leq T\leq \mathds{1}$
such that
$T$ is not orthogonal to $\rho$.
\end{lemm1}

\begin{proof}
We first consider the case of $\alpha>1$ and show \eqref{eq:variationa_sandwiched_app} for finite-dimensional $\mathcal{H}$ to present the core idea.
The case for separable infinite-dimensional $\mathcal{H}$ follows from finite-dimensional approximations (see e.g.~\cite[\S 3.3]{Mos23}).
For any bounded positive semi-definite operator $T\geq0$ that is not orthogonal to $\rho$, we have $\Tr[\rho T] \in (0,+\infty)$.
By letting $X = \frac{\rho}{\sigma}$ and $Y = T$,
H\"older's inequality (Theorem~\ref{theo:Holder}) implies
\begin{align}
	\label{eq:change_of_measure_sandwiched}
	\Tr[\rho T]
	&=\left\langle \frac{\rho}{\sigma}, T \right\rangle_{\sigma}
	\\
	\label{eq:Holder_sandwiched}
	&\leq \left\| \frac{\rho}{\sigma} \right\|_{\alpha, \sigma } \cdot \left\| T \right\|_{\alpha', \sigma }
	\\
	&= \e^{ \frac{1}{\alpha'}  D_{\alpha}^*(\rho\Vert \sigma) } \cdot \left( \Tr\left[ \left( \sigma^{\frac{1}{2\alpha'}} T \sigma^{\frac{1}{2\alpha'}} \right)^{\alpha'} \right] \right)^{\frac{1}{\alpha'}},
\end{align}
which translates to
\begin{align} \label{eq:final_sandwiched}
	D_\alpha^*(\rho \Vert \sigma )
	\geq \left\{ \alpha' \log \Tr\left[\rho T\right]
	- \log \Tr\left[ \left(  \sigma^{\frac{1}{2\alpha'}}  T \sigma^{\frac{1}{2\alpha'}}  \right)^{\alpha'} \right] \right\}, \quad \forall\, T\geq 0, \, T\not\perp \rho.
\end{align}
Given $D_{\alpha}^*(\rho\Vert \sigma)
=\alpha'\log\| \sigma^{-\frac{1}{2\alpha'}}\rho \sigma^{-\frac{1}{2\alpha'}}\|_{\alpha}<+\infty$, the inequality \eqref{eq:Holder_sandwiched} can be attained by
\begin{align} 
	T=\sigma^{-\frac{1}{2\alpha'}} \left|\sigma^{-\frac{1}{2\alpha'}}\rho \sigma^{-\frac{1}{2\alpha'}}\right|^{\al-1} \sigma^{-\frac{1}{2\alpha'}}\in L_{\al'}(\mathcal{B(H)},\sigma)\pl.
\end{align}

Here, the quotient $X = \frac{\rho}{\sigma}$ can be considered as a \emph{noncommutative likelihood ratio}, and the equality \eqref{eq:change_of_measure_sandwiched} is a kind of change of measure argument from $\rho$ to the quotient $\frac{\rho}{\sigma}$. 
The restriction to $T\le \mathds{1}$ follows from the fact that the variational expression takes same value for $T$ and $\lambda T$ for any $\lambda>0$; we can hence normalize $T$ by its operator norm.

Similarly for $\alpha \in (0,1)$, we apply the reverse H\"older's inequality and restrict to $T>0$.
Then the inequality of \eqref{eq:Holder_sandwiched} becomes reversed.
In the end, we still arrive at \eqref{eq:final_sandwiched} because $\al'=\frac{\alpha}{\alpha-1} < 0$. 

Since $\sigma$ is a density operator, then
by the density of $\mathcal{B(H)}\subset L_{\al'}(\mathcal{B(H)},\sigma)$, the inequality \eqref{eq:Holder_sandwiched} can also be approximately saturated by bounded $T\in \mathcal{B(H)}$, which concludes the variational expression of the sandwiched quantity $D^*(\rho \Vert \sigma )$ in \eqref{eq:variationa_sandwiched_app}.

\medskip

We move on to the measured quantity $D_{\alpha}^{\mathds{M}}(\rho \Vert \sigma )$ and $\al>1$. We start with a simple case that $\mathcal{H}$ is finite dimensional so that we can assume $T\geq0$ that admits an finite orthogonal decomposition:
$ T = \sum_i t(i) {P}(i)$, where $t = \left(t(i)\right)_i$ is a sequence of non-negative numbers, and ${P} = \left(P(i)\right)_i$ is a collection of mutually orthogonal projections summing to identity. 
The infinite-dimensional case follows from, again, finite-dimensional approximations and the fact that the clasical R\'enyi divergence can be approximated from finite partitions.

We further denote the induced probability mass functions $	p_{\rho,{P}}$ and $	q_{\sigma,{P}}$ under the projective measurement $P$:
\begin{align} \label{eq:distributions}
	p_{\rho,{P}}(i) := \Tr\left[ \rho {P}(i) \right],
	\quad
	q_{\sigma,{P}}(i) := \Tr\left[ \sigma {P}(i) \right], \quad \forall\, i
\end{align}
Then, we apply the same reasoning of change of measure as \eqref{eq:change_of_measure_sandwiched} and (commutative) H\"older's inequality (Theorem~\ref{theo:Holder}) to obtain
\begin{align}
	\label{eq:change_of_measure_measured}
	\Tr[\rho T]
	&=\left\langle \frac{ p_{\rho,{P}} }{ q_{\sigma,{P}} }, t \right\rangle_{ q_{\sigma,{P}} }
	\\
	\label{eq:Holder_measured}
	&\leq \left\| \frac{ p_{\rho,{P}} }{q_{\sigma,\mathrm{P}} } \right\|_{\alpha, q_{\sigma,{P}} } \cdot \left\| t \right\|_{\alpha', q_{\sigma,{P}} }
	\\ \nonumber
	&= \e^{ \frac{\alpha-1}{\alpha}  D_{\alpha}\left( p_{\rho,\mathrm{P}} \Vert q_{\sigma,{P}} \right) } \cdot \left( \sum\nolimits_i q_{\sigma,{P}}(i) t_i^{\alpha'} \right)^{\frac{1}{\alpha'}}
	\\ \nonumber
	&= \e^{ \frac{\alpha-1}{\alpha}  D_{\alpha}\left( p_{\rho,{P}} \Vert q_{\sigma,{P}} \right) } \cdot \left( \Tr\left[ \sigma T^{\alpha'} \right] \right)^{\frac{1}{\alpha'}}.
\end{align}
The above then translates to
\begin{align} \label{eq:temp}
	  D_{\alpha}\left( p_{\rho,{P}} \Vert q_{\sigma,{P}} \right)
	\geq \left\{ \alpha' \log \Tr\left[\rho T\right]
	- \log \Tr\left[ \sigma  T^{\alpha'}  \right] \right\}, 
\end{align}
which is saturated by choosing sequence $t= \left( \frac{p_{\rho,P}}{q_{\sigma,P}} \right)^{\alpha-1} $ and further maximizing over all finite PVMs $\sum_{i}P(i)=\mathds{1}$, i.e.
\begin{align*} \sup_{T\geq 0} \left\{ \alpha' \log \Tr\left[\rho T\right]
	- \log \Tr\left[ \sigma  T^{\alpha'}  \right] \right\}
	&=\sup_{\sum_i P(i)=\mathds{1}}\sup_{t_i\ge 0}\left\{ \alpha' \log \Tr\left[\rho T\right]
	- \log \Tr\left[ \sigma  T^{\alpha'}  \right] \right\}\\
&=\sup_{\sum_i P(i)=\mathds{1}} D_{\alpha}\left( p_{\rho,{P}} \Vert q_{\sigma,{P}} \right)\\
&= D_{\alpha}^{\mathds{P}}(\rho \Vert \sigma )
\\
&=D_{\alpha}^{\mathds{M}}(\rho \Vert \sigma ).
\end{align*}
The case of $\alpha\in (0,1)$ is similar using reverse H\"older inequality.

For the case of separable Hilbert spaces $\mathcal{H}$, it suffices to consider positive $T=\sum_{i} t(i) P(i)$ with a finite orthogonal decomposition. 
Take bounded $T$ such that
\begin{align} \alpha' \log \Tr\left[\rho T\right]
	- \log \Tr\left[ \sigma  T^{\alpha'}  \right] >D_{\alpha}^{\mathds{P}}(\rho\Vert \sigma)-\epsilon.
\end{align}
Since $\text{supp}(\rho) \subseteq  \text{supp}(\sigma)$ and $T\not\perp \rho$, again by homogeneity we can assume $0\leq T\le\mathds{1} $ satisfying,
\begin{align} 
	0<\tr[\rho T]\le 1, \quad  0<\tr\left[\sigma T^{\alpha'} \right]\le 1\pl.
\end{align}
Let $T=\int_{0}^1 t d{E_t}$ be its spectrum decomposition in the integral form. By function calculus, one can take $T_n=\sum_{i=0}^{n-1} \frac{i}{n}E_{[\frac{i}{n},\frac{i+1}{n})}$. 
By the following convergence in operator norm:
\begin{align}
\norm{T-T_n}{}\le \frac{1}{n}\pl, \quad
\left\|T^{t}-T_n^t\right\| \le \max\left\{\frac{1}{n^t}, 1-\left(\frac{n-1}{n}\right)^t\right\},
\quad \forall\, t\in(0,+\infty), 
\end{align}
we have
\begin{align}
\Tr\left[\rho T\right]=\lim_{n\to \infty} \Tr\left[\rho T_n\right]\pl,\pl \Tr\left[ \sigma  T^{\alpha'}\right]=\lim_{n\to \infty}\Tr\left[ \sigma  T_n^{\alpha'}\right].
\end{align}
Then
\begin{align}
\sup_{T\geq 0}
	\left\{ \alpha' \log \Tr\left[\rho T\right]
	- \log \Tr\left[ \sigma  T^{\alpha'}  \right] \right\}
	&=\sup_{T=\sum_{i=1} t(i) P(i)}
	\left\{ \alpha' \log \Tr\left[\rho T\right]
	- \log \Tr\left[ \sigma  T^{\alpha'}  \right] \right\}
	\\
	&=\sup_{\sum_{i}P(i)=\mathds{1}}
	D_{\alpha}\left( p_{\rho,{P}} \Vert q_{\sigma,{P}} \right).
\end{align}
where the second supremum is over all positive $T$ with finite spectrum decomposition, and the third supremum is over PVMs finite outcomes. Then the assertion follows from the next lemma.
\end{proof}

\begin{lemm}Let $\mathcal{H}$ be a separable Hilbert space. For $\al\in (0,+\infty]$ and any two density operators $\rho $ and $\sigma$,
\[ D_{\alpha}^{\mathds{M}}(\rho \Vert \sigma )=D_{\alpha}^{\mathds{P}}(\rho \Vert \sigma )=\sup_{\sum_{i}P(i)=\mathds{1}} D_{\alpha}\left( p_{\rho,{P}} \Vert q_{\sigma,{P}} \right)
\]
where the supremum is over all PVM $\sum_{i=1}^kP(i)=\mathds{1}$ with finite outcome.
The induced probability distributions $p_{\rho,{P}} $ and $ q_{\sigma,{P}}$
are given in \eqref{eq:distributions}.
\end{lemm}
\begin{proof}Recall that a general PVM on separable $H$ is a map
\[ \Phi: \mathcal{A} \to \mathcal{B(H)}\pl,\]
from a Borel $\sigma$-algebra $\mathcal{A}$ over some measure space $\Omega$ satisfying the following properties,
\begin{itemize}
\item[i)] for all $A\in \mathcal{A}$, $\Phi(A)$ is a projection operator in $\mathcal{B(H)}$.
\item[ii)] $\Phi(\varnothing)=0$ and $\Phi(\Omega)=\mathds{1}$.
\item[iii)] for disjoint $A_1,A_2\in \mathcal{A}$, $\Phi(A_1)\Phi(A_2)=0$ .
\item [iv)] for a countable family of mutually disjoint sets $A_1,A_2,\cdots$,
\[ \Phi( \cup_{i}^\infty A_i)=\sum_{i=1}^\infty\Phi(A_i)\pl.\]
\end{itemize}
The linearization of this map induces a normal map (which we also denote by $\Phi$)
\[ \Phi:L_\infty(\Omega, \mathcal{A})\to \mathcal{B(H)} \]
whose predual
\[ \Phi_*: \mathcal{S}_1(\mathcal{H})\to  (L_\infty(\Omega, \mathcal{A}))_*\]
is a map sending a density operator (in the Schatten $1$-class $\mathcal{S}_1(\mathcal{H})$) to a probability measure on $\Omega$. In general, the projectively-measured R\'enyi divergence
\[D_{\alpha}^{\mathds{P}}(\rho\Vert \sigma)=\sup_{\Phi \text{ PVM}} D_\alpha(\Phi_*(\rho)\Vert\Phi_*(\sigma)) \]
should consider the supremum over all general PVMs $\Phi:\mathcal{A}\to B(H)$. 
Here $\Phi_*(\rho)$, $\Phi_*(\sigma)$ are two probability distributions on $\Omega$. 

Now we recall that
for any two probability distributions $\mu$ and $\nu$,
\[ D_\alpha(\mu\Vert\nu)=\sup_{\{A_i\}_i} D( \mu_A \Vert \nu_A )\]
where the supremum is over any finite partition of $\Omega$ \cite[Theorem 2]{EH14}.  Here $\mu_A $ is the finite probability density with $\mu_A(i)=\mu(A_i) $ and similarly for $\nu_A$.
That is, the $D_\alpha$ for two probability measures on any type of measure spaces can be approximated by finite outcome coarse graining. Thus we have
\begin{align*}
	D_{\alpha}^{\mathds{P}}(\rho\Vert \sigma)&=\sup_{\Phi \text{ PVM}} D_\alpha(\Phi_*(\rho)\Vert\Phi_*(\sigma))
	\\
	&=\sup_{(A_i)_i\,\text{finite  partition}}\sup_{\Phi \text{ PVM}} D_\alpha(\Phi_*(\rho)_A\Vert\Phi_*(\sigma)_{A})
	\\
	&\le \sup_{\sum_{i}P(i)=\mathds{1}}D_\alpha(p_{\rho,{P}}\Vert q_{\sigma,{P}}) .
\end{align*}
Here, the last inequality follows from that
\[\Phi_*(\rho)_A(i)
=\Phi_*(\rho)(A_i)=\tr\left[\rho\Phi(A_i)\right]\pl, 
\quad \Phi_*(\sigma)_{A}(i)=\Phi_*(\sigma)(A_i)=\tr\left[\sigma\Phi(A_i)\right]
\]
are the measurement outcome of the finite PVM $\sum_{i}\Phi(A_i)=\mathds{1}$.
\end{proof}

\begin{remark}
Mosonyi established the variational expression for the so-called $(\alpha,z)$-R\'enyi divergence \cite{AD15} also using H\"older's inequality \cite[Lemma 3.23]{Mos23}.
Note that $(\alpha,\alpha)$-R\'enyi divergence coincides with the sandwiched R\'enyi divergence $D_{\alpha}^*$.
Our proof for \eqref{eq:variationa_sandwiched_app} employs H\"older's inequality in a slightly different way. Yet it yields a similar proof for the measured quantity $D_{\alpha}^{\mathds{M}}$.

Mosonyi's result also provid a variational expression for the Petz--R\'enyi divergence \cite{Pet86} (which with abuse of notation we still denote by $D_{\alpha}(\rho\Vert\sigma) := \frac{1}{\alpha-1}\log \Tr[ \rho^{\alpha} \sigma^{1-\alpha} ]$) as follows:
\begin{align}
		D_{\alpha}(\rho\Vert \sigma)
		&= \sup_{0<T\leq \mathds{1}} \left\{
		\frac{\alpha}{\alpha-1} \log \Tr\left[ \left( T^{\frac{\alpha}{2}} \rho^\alpha T^{\frac{\alpha}{2}} \right)^{\frac{1}{\alpha} } \right]
		- \log \Tr\left[ \left( T^{\frac{\alpha}{2}} \sigma^{\alpha-1} T^{\frac{\alpha}{2}}  \right)^{\frac{1}{\alpha-1}}
		\right]
		\right\}.
	\end{align}
The proof also follows from H\"older inequality, e.g.~for $\alpha>1$,
\begin{align}
		\Tr\left[ \left( T^{\frac{\alpha}{2}} \rho^\alpha T^{\frac{\alpha}{2}} \right)^{\frac{1}{\alpha} } \right]=\left\| \rho^{\frac{\alpha}{2}} T^{\frac{\alpha}{2}} \right\|_{\frac{2}{\alpha}}^2
		=\left\| \rho^{\frac{\alpha}{2}} \sigma^{\frac{1-\alpha}{2}} \sigma^{\frac{\alpha-1}{2}} T^{\frac{\alpha}{2}} \right\|_{\frac{2}{\alpha}}^2
		&\leq \left\| \rho^{\frac{\alpha}{2}} \sigma^{\frac{1-\alpha}{2}} \right\|_2^2
		\cdot \left\| T^{\frac{\alpha}{2}} \sigma^{\frac{\alpha-1}{2}} \right\|_{\frac{2}{\alpha-1}}^2 \nonumber
		\\
		&= \Tr\left[ \rho^\alpha \sigma^{1-\alpha} \right] \cdot \Tr\left[ \left( T^{\frac{\alpha}{2}} \sigma^{\alpha-1} T^{\frac{\alpha}{2}}  \right)^{\frac{1}{\alpha-1}}
		\right].
	\end{align}
\end{remark}

{\larger
\bibliographystyle{myIEEEtran}
\bibliography{reference, Hao-Chung}
}

\end{document}